\documentclass[10pt, conference]{IEEEtran}

\usepackage{graphicx}
\usepackage{amssymb,amsmath}
\usepackage{amsthm}
\usepackage{textcomp}       
\usepackage{indentfirst} 
\usepackage{url}
\usepackage{color}
\usepackage{marvosym}
\usepackage{algorithmic}
\usepackage[skip=10pt]{caption}
\usepackage{tabularx}
\usepackage{multirow}
\usepackage{hhline}
\usepackage[table]{xcolor}
\usepackage{cite}
\usepackage{setspace}
\usepackage{afterpage}
\usepackage{enumitem}
\usepackage{braket}
\usepackage{bbm}
\usepackage{bm}
\usepackage{epstopdf}
\usepackage{tcolorbox}
\usepackage{booktabs}
\usepackage{graphicx}
\usepackage{amssymb}
\usepackage{bm}
\usepackage{epstopdf}
\usepackage{subcaption}
\newcommand{\norm}[1]{\left\lVert#1\right\rVert}

\def\R{\mathbb{R}}

\usepackage[ruled,commentsnumbered]{algorithm2e}

\newcounter{subeqn} %
\makeatletter
\@addtoreset{subeqn}{equation}
\makeatother

\newtheorem{theorem}{Theorem}

\newtheorem{assumption}{Assumption}

\usepackage{soul}

\def\1{\mathbf{1}}

\def\BibTeX{{\rm B\kern-.05em{\sc i\kern-.025em b}\kern-.08em T\kern-.1667em\lower.7ex\hbox{E}\kern-.125emX}}

\title{Improving IoT Analytics through Selective Edge Execution}

\author{\IEEEauthorblockN{Apostolos Galanopoulos\IEEEauthorrefmark{1}, Argyrios G. Tasiopoulos\IEEEauthorrefmark{2}, George Iosifidis\IEEEauthorrefmark{1}, \\ Theodoros Salonidis\IEEEauthorrefmark{3}, Douglas J. Leith\IEEEauthorrefmark{1}}\\
	\IEEEauthorblockA{
		\IEEEauthorrefmark{1}School of Computer Science and Statistics, Trinity College Dublin\\
		\IEEEauthorrefmark{2}Department of Electronic and Electrical Engineering, University College London\\
		\IEEEauthorrefmark{3}IBM T. J. Watson Research Center, New York}
}

\vspace{-3mm}

\IEEEoverridecommandlockouts
\begin{document}
\maketitle
\thispagestyle{plain}
\pagestyle{plain}

\begin{abstract}
A large number of emerging IoT applications rely on machine learning routines for analyzing data. Executing such tasks at the user devices improves response time and economizes network resources. However, due to power and computing limitations, the devices often cannot support such resource-intensive routines and fail to accurately execute the analytics. In this work, we propose to improve the performance of analytics by leveraging edge infrastructure. We devise an algorithm that enables the IoT devices to execute their routines locally; and then outsource them to cloudlet servers, only if they predict they will gain a significant performance improvement. It uses an approximate dual subgradient method, making minimal assumptions about the statistical properties of the system's parameters. Our analysis demonstrates that our proposed algorithm can intelligently leverage the cloudlet, adapting to the service requirements. 

\end{abstract}

\vspace{1mm}

\section{Introduction}

The recent demand for machine learning (ML) applications, such as image recognition, natural language translation, and health monitoring, has been unprecedented \cite{tiropanis-survey}. These services collect data streams generated by small devices, and analyze them locally or at distant cloud servers. There is growing consensus that such applications will be ubiquitous in Internet of Things (IoT) systems \cite{Jiang-17}. The challenge, however, with such services is that they are often resource intensive. On the one hand, the cloud offers powerful ML models and abundant compute resources but requires data transfers which consume network bandwidth and might induce significant delays \cite{cisco-cloud}. On the other hand, executing these services at the devices economizes bandwidth but degrades their performance due to the devices' limited resources, e.g. memory or energy.

A promising approach to tackle this problem is to allow the devices to outsource individual ML tasks to edge infrastructure such as cloudlets \cite{Cloudlets}. This can increase their execution accuracy since the cloudlet's ML components are typically more complex, and hence offer improved results. Nevertheless, the success of such solutions presumes intelligent outsourcing algorithms. The cloudlets, unlike the cloud, have limited computing capacity and cannot support all requests. At the same time, task execution requires the transfer of large data volumes (e.g., video streams). This calls for prudent transmission decisions in order to avoid wasting device energy and bandwidth. Furthermore, unlike prior computation offloading solutions \cite{letaief-edge-tutorial}, it is crucial to only outsource the tasks that can significantly benefit from cloudlet execution. 

\emph{Our goal is to design an online framework that addresses the above issues and makes intelligent outsourcing decisions}. We consider a system where a cloudlet improves the execution of image classification tasks running on devices such as wireless IoT cameras. 
We assume that each device has a "low-precision" classifier while the cloudlet can execute the task with higher precision. The devices classify the received objects upon arrival, and decide whether to transmit them to the cloudlet or not, to get a better classification result. Making this decision requires an assessment of the potential performance gains, which are measured in terms of accuracy improvements. To this end, we propose the usage of a \emph{predictor} at each device that leverages the local classification results.

We consider the practical case where the resources' availability is unknown and time-varying, but their instantaneous values are observable. We design a distributed adaptive algorithm that decides the task outsourcing policy towards maximizing the long-term performance of analytics. To achieve this, we formulate the system's operation as an optimization problem, which is decomposed via Lagrange relaxation to a set of device-specific problems. This enables its distributed solution through an \emph{approximate} -- due to the unknown parameters -- dual ascent method, that can be applied in real time. The method is inspired by primal averaging schemes for \emph{static} problems, e.g., see \cite{nedic-subgrad-siam}, and achieves a \emph{bounded and tunable} optimality gap using a novel approximate iteration technique. Our contributions can be summarized as follows:
\begin{itemize}[leftmargin=3mm]
\item \underline{\emph{Edge Analytics}}. We study the novel problem of intelligently improving data analytics tasks using edge infrastructure, which is increasingly important for the IoT. 

\item \underline{\emph{Decision Framework}}. We propose an \emph{online} task outsourcing algorithm that achieves near-optimal performance under very general conditions (unknown, non i.i.d. statistics). This is a novel analytical result of independent value.

\item \underline{\emph{Implementation \& Evaluation}}. The solution is evaluated in a wireless testbed using a ML application, several classifiers and datasets. We find that our algorithm increases the accuracy (up to $32\%$) and reduces the energy (down to $60\%$) compared to carefully selected benchmark policies.
\end{itemize}

\textbf{Organization}. Sec. \ref{sec:model} introduces the model and the problem. Sec. \ref{sec:algorithms} presents the algorithm and Sec. \ref{sec:evaluation} the system implementation, experiments and trace-driven simulations. We discuss related work in Sec. \ref{sec:related} and conclude in Sec. \ref{sec:conclusions}. 
Although the paper is completely self-sufficient, the interested reader will find more results from the implementation of our system, as well as a more detailed version of the proof of our main analytical contribution in \cite{appendix}.
\section{Model and Problem Formulation} \label{sec:model}

\textbf{Classifiers}. There is a set $\mathcal{C}$ of $C$ disjoint object classes and a set $\mathcal{N}$ of $N$ edge devices. We assume a time-slotted operation where each device $n$ receives at slot $t$ a group of objects (or tasks) $\mathcal{S}_{nt}$ to be classified, e.g., frames captured by its camera. We define $\mathcal{S}_n\!\supseteq\! \mathcal{S}_{nt}, \forall t$ as the set of objects that can arrive at $n$, and $\mathcal{S}\!=\!\cup\{\mathcal S_{n}\}_{n}$. Each device $n$ is equipped with a \emph{local} classifier $J_n\!:\!\mathcal{S}_{n}\!\rightarrow\! \big(\mathcal{C}_n, d_{n}(s_{nt}) \big)$, which outputs the inferred class of an object $s_{nt}$ and a normalized confidence value $d_n(s_{nt})\in[0,1]$ for that inference\footnote{The classifier might output only the class with the highest confidence, or a vector with the confidence for each class; our analysis holds for both cases.}. The cloudlet has a classifier $J_0\!:\!\mathcal{S}\!\rightarrow\!\big(\mathcal{C}_0, d_{0}(s_{nt}) \big)$ that can classify any object, and offers higher accuracy from all devices, i.e., $d_0(s_{nt})\!\geq\! d_n(s_{nt}), \forall n\in\mathcal{N}$.

Let $\phi_{nt}\in[0,1]$ denote the accuracy improvement when the cloudlet classifier is used:
\begin{equation}
	\phi_{nt}(s_{nt})=d_{0}(s_{nt}) - d_n(s_{nt}),\,\,\,\,\, \forall n\in\mathcal{N},\, s_{nt}\in\mathcal{S}_{nt}. \label{eq:predictor-gain}
\end{equation}
Every device is also equipped with a predictor\footnote{This can be a model-based or model-free solution, e.g., a regressor or a neural-network; our analysis and framework work for any of these solutions. In the implementation we used a mixed-effects regressor, see \cite{regression-multilevel}. } $Q_n$ that is trained with the outcomes of the local and cloudlet classifiers. This predictor can estimate the accuracy improvement offered by the cloudlet for each object $s_{nt} \in \mathcal{S}_{nt}$:
\begin{equation}
 Q_n: \big(J_n(s_{nt})\big) \rightarrow \big(\hat{\phi}_{nt},\, \sigma_{nt}\,\big),
\end{equation} 
and, in general, this assessment might be inexact, $\hat{\phi}_{nt}(s_{nt})\!\neq\! \phi_{nt}(s_{nt})$, and $\sigma_{nt}\!\in\![0,1]$ is the respective confidence value.

\textbf{Wireless System}. The devices access the cloudlet through high capacity cellular or Wi-Fi links. Each device $n$ has an \emph{average power budget} of $B_n$ Watts. Power is a key limitation here because the devices might have a small energy budget due to protocol-induced transmission constraints, or due to user aversion for energy spending. The cloudlet has an \emph{average processing capacity} of $H$ cycles/sec which is shared by the devices, and when the total load exceeds $H$, the task delay increases and eventually renders the system non-responsive.

\emph{We consider the realistic scenario where the parameters of devices and the cloudlet change over time in an unknown fashion}. Namely, they are created by random processes $\{B_{nt}\}_{t=1}^{\infty}$ and $\{H_t\}_{t=1}^{\infty}$, and our decision framework has access only to their instantaneous values in each slot. Unlike previous optimization frameworks \cite{tassiulas-book} that assume i.i.d., or Markov modulated processes; here we only ask that these perturbations are bounded in each slot, i.e. $H_{t}\leq H_{max},\,B_{nt}\leq B_{max}, \forall t$ and their averages converge to some finite values which we do need to know, i.e., $\lim_{t\rightarrow \infty}\sum_{\tau=1}^tB_{nt}/t=B_n,\forall n$, and similarly for $\{H_t\}_{t=1}^\infty$. We also define $\bm{B}_t=(B_{nt}, n\in\mathcal{N})$.

When an object (say, image) is transmitted in slot $t$ from device $n$ to the cloudlet, it consumes\footnote{Power budgets are also affected by the local classifier computations which are made for every object and thus do not affect the offloading decisions.} part of the device's power budget $B_n$. We assume that this cost, denoted $o_{nt}$, follows a random process $\{o_{nt}\}_{t=1}^{\infty}$ that is uniformly upper-bounded and has well-defined mean values.\footnote{This cost can reflect, e.g., the impact of time-varying channel conditions.} Also, each transmitted object requires a number of processing cycles in the cloudlet which might also vary with time, e.g., due to the different type of the objects, and we assume it follows the random process $\{h_{nt}\}_{t=1}^{\infty}$, with $\lim_{t\rightarrow \infty}\sum_{\tau=1}^t h_{nt}/t=h_n$. We define $\bm{o}_t\!=\!(o_{nt}\!\leq\! o_{max}, n\in\mathcal{N})$, and $\bm{h}_t\!=\!(h_{nt}\!\leq\! h_{max}, n\in\mathcal{N})$. Our model is very general as the \emph{(i)} requests, \emph{(ii)} power and computing cost per request, and \emph{(iii)} resource availability, can be arbitrarily time-varying, and with unknown statistics.
 
\begin{figure}[t]
	\centering
	\includegraphics[scale=0.37]{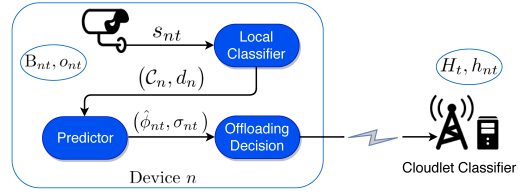}
	\caption{Schematic of the basic notation and procedure followed by the system's devices.}
	\label{fig:system-model}
	\vspace{-2mm}
\end{figure}

\textbf{Problem Formulation}. The IoT devices wish to involve the cloudlet only when they confidently expect high classification precision gains. Otherwise, they will consume the cloudlet's capacity and their own power without significant performance benefits. Therefore, we make the outsourcing decision for each object $s_{nt}$ based on the \emph{weighted improvement gain}:
\begin{equation} \label{eq:accuracy}
w_{nt}(s_{nt})=\hat{\phi}_{nt}-\rho_{n}\sigma_{nt},\,\forall\,n,t\,,
\end{equation}
where $\rho_{n}\!\geq\!0$ is a \emph{risk aversion} parameter set by the system designer or each user. For example, assuming normal distribution for $\phi_{nt}$, we could set $\rho_{n}\!=\!1$ and use a threshold rule of $1$ standard deviation. We use hereafter these modified parameters $w_{nt}, \forall n$, and partition the interval of their values $[-w_{0},w_{0}]$ ($w_0$ being the maximum) into subintervals $\mathcal{I}^j, j=1,\ldots,M$ such that $\cup_{j=1}^{M}\mathcal{I}_j=[-w_{0},w_{0}], \forall i\!\neq\! j$; with $w_n^j$ being the center point of $\mathcal{I}^j$. This quantization facilitates the implementation of our algorithm in a real system, and is without loss of generality since we can use very short intervals. Finally, let $\lambda_{nt}^j$ denote the number of objects with expected gain $w_n^j$ that device $n$ has created in slot $t$. These \emph{arrivals} are generated by an unknown process $\{\lambda_{nt}^j \}_{t=1}^{\infty}$, with $\lim_{T\rightarrow \infty}1/T\sum_{t=1}^{T}\lambda_{nt}^j=\lambda_n^j, \forall n,j$. 


\emph{Our aim is to maximize the aggregate long-term analytics performance gains, for all objects and IoT devices.} This can be formulated as a mathematical program. We define variables $y_{n}^j\in[0,1], \forall n,j$ which indicate the long term ratio of objects with expected gain of $w_n^j$ that are sent to the cloudlet (with $y_n^j\!=\!1$, when all objects of $n$ in $\mathcal I^j$ are sent), and formulate the convex problem:
\vspace{-1mm}
\begin{subequations}
\begin{align}
\mathbb P_1:&\qquad\underset{y_{n}^j\in [0,1]}{\text{maximize}} \sum_{j=1}^{M}\sum_{n=1}^N w_n^j\lambda_{n}^jy_{n}^j\triangleq f(\bm{y})\\
s.t.&\qquad \sum_{j=1}^M y_{n}^j\lambda_{n}^jo_{n} \le B_n,\,n\in\mathcal{N}, \label{eq:P1-device-constraint}\\
& \qquad \sum_{j=1}^M\sum_{n=1}^N y_{n}^j\lambda_{n}^jh_{n}\le H, \label{eq:P1-cloudlet-constraint}
\end{align}
\end{subequations}
where $\bm{y}\!=\!( y_{n}^j\!:\! \forall n, j )$. Eq. \eqref{eq:P1-device-constraint} constraints the average power budget of each device and \eqref{eq:P1-cloudlet-constraint} bounds the cloudlet utilization. Clearly, based on the specifics of each system we can add more constraints, e.g., for the average wireless link capacity in case bandwidth is also a bottleneck resource. Such extensions are straightforward as they do not change the properties of the problem, nor affect our analysis below. 

The solution of $\mathbb P_1$ is a policy $\bm y^*$ that maximizes the aggregate (hence also average) analytics performance in the system. Such policies can be randomized, with $y_{n}^{j*}$ denoting the probability of sending each object of $n$ in interval $\mathcal I^j$ to the cloudlet (at each slot). However, in reality, the system parameters not only change with time, but are generated by processes that might not be i.i.d. and have unknown statistics (mean values, etc.). This means that in practice we cannot find $\bm y^*$. In the next section we present an online policy that is oblivious to the statistics of $\{\bm{\lambda}_t \}, \{\bm{o}_t\}, \{\bm{h}_t\}, \{H_t\}, \{\bm{B}_t\}$ but achieves indeed the same performance with $\bm y^*$.

%
%
%
%
%
%
%
%
%
\section{Online Offloading Algorithm} \label{sec:algorithms}

Our solution approach is simple and, we believe, elegant. We replace the unknown parameters $H$, $\lambda_n^j$, $B_n$, $o_n$ and $h_n$, $\forall n,j$ in $\mathbb P_1$ with their running averages (which we calculate as the system operates), solve the modified problem with gradient ascent in the dual space, and perform primal averaging. This gives us an online policy that applies in real time the solution $\bm y_t, \forall t,$ while using only information made available by slot $t$. 

\vspace{-1mm}
\subsection{Problem Decomposition \& Algorithm Design}
Let us first define the running-average function:
\begin{align*}
	\bar f_t(\bm y)&\triangleq\sum_{j=1}^{M}\sum_{n=1}^N w_n^jy_{n}^j\bar \lambda_{nt}^j\\ \notag
	&=\sum_{j=1}^{M}\sum_{n=1}^N w_n^j\lambda_{n}^jy_{n}^j+\sum_{j=1}^{M}\sum_{n=1}^N w_n^jy_{n}^j(\lambda_{n}^j-\bar \lambda_{nt}^j) \\ \notag
	&=f(\bm y) + \bm y^\top\epsilon_t, \nonumber
\end{align*}
where $\bar \lambda_{nt}^j\!=\!\sum_{\tau=1}^t\lambda_{n\tau}^j/t$ is the running average of $\lambda_{n}^j$, and $\epsilon_t\!=\!\big(w_{n}^j(\lambda_n^j\!-\!\bar \lambda_{nt}^j\big), \forall n, j ) \in \R^{NM}$ is the vector of component-wise errors between $\bar f_t(\bm y)$ and $f(\bm y)$. Also, we denote $g(\bm{y})\!\in\!\R^{N+1}$ the constraint vector of \eqref{eq:P1-device-constraint}-\eqref{eq:P1-cloudlet-constraint}, and define
\begin{equation} 
\bar g_t(\bm y)=g\big(\bm y) + \delta_t(\bm y\big),
\end{equation}
with $\delta_t(\bm y)=\big(\delta_{nt}(\bm y), n=1,\ldots,N+1\big)$ and
\begin{align*}
&\delta_{nt}(\bm y)\!=B_n\!-\bar B_{nt}\!+\sum_{j=1}^M\!y_{n}^j\big( \bar o_{nt}\bar \lambda_{nt}^j - o_n \lambda_{n}^j \big), n=1,\ldots,N,
\end{align*}
\begin{align*}
&\delta_{{N+1,t}}(\bm y)=H\!-\bar H_{t}+\sum_{j=1}^M\sum_{n=1}^Ny_n^j\big(\bar h_{nt}\bar \lambda_{nt}^j - h_n\lambda_n^j \big). 
\end{align*}
$\bar B_{nt}\!=\!\sum_{\tau=1}^{t}B_{n\tau}/t$ is the running average of process $\{B_{nt}\}_{t=1}^{\infty}$, and similarly we define $\bar H_t$, $\bar o_{nt}$, and $\bar h_{nt}$. Note that $\bar f_t(\bm y), \bar g_t(\bm y)$ can be calculated at each slot, while $f(\bm y)$ and $g(\bm y)$ are unknown. We can now define a new problem:
\begin{equation}
\mathbb P_{2}(t):\,\,\,	\max_{\bm y\in [0,1]^{NM}} \bar f_t(\bm y) \,\,\,\,s.t.\,\,\,\, \bar g_t(\bm y)\preceq 0 \notag
\end{equation}
We will use the instances $\{ \mathbb P_2(t) \}_t$ to perform a dual ascent method and obtain a sequence of decisions $\{\bm y\}_t$ that will be applied in real time and achieve performance that converges asymptotically to the (unknown) solution of $\mathbb P_1$.

We first dualize $\mathbb P_2(t)$ and introduce the Lagrangian\footnote{For our system implementation, this relaxation means we install queues for the data transmission (at the devices) and image processing (at the cloudlet).}: 
\begin{align}
&L(\bm{y},\bm{\mu})\triangleq\bar f_t(\bm y) +  \bm{\mu}^\top \bar g_t(\bm y)=  \sum_{j=1}^{M}\sum_{n=1}^N w_n^jy_{n}^j\bar \lambda_{nt}^j + \notag\\
&\sum_{n=1}^N\mu_n\big( \sum_{j=1}^My_n^j\bar\lambda_{nt}^j\bar o_{nt}-\bar B_{nt} \big) + \xi\big( \sum_{j=1}^M\sum_{n=1}^Ny_{n}^j\bar \lambda_{nt}^j\bar h_{nt}-\bar H_t  \big) \notag
\end{align}
where $\bm{\mu}\!=\!(\mu_1, \mu_2, \ldots,\mu_N,\xi)$ are the non-negative dual variables for $\bar g_t(\bm y)\!\preceq\!0$. The dual function is:
\begin{equation}
V(\bm{\mu})=\arg\min_{\bm 0\preceq\bm{y}\preceq\bm 1}L(\bm{y},\bm{\mu}),  \label{eq:dual-calc}
\end{equation}
and the dual problem amounts to maximizing $V(\bm{\mu})$. 

We apply a dual ascent algorithm where the iterations are in sync with the system's time slots $t$. Observe that $V(\bm{\mu})$ does not depend on $\bar B_{nt}$ or $\bar H_t$, it is separable with respect to the primal variables, and independent of $\bar \lambda_{nt}^j$. Hence, in each iteration $t$ we can minimize $L$ by:
\begin{equation}
(y_{n}^{j})^*\!\in\!\arg\!\!\min_{y_{n}^{j}\in[0,1]}y_{n}^j(-w_{n}^j\!+\!\mu_{nt}\bar o_{nt}\! +\!\xi_{t} \bar h_{nt}),\,\, \forall\, n, j. \label{eq:sub0}
\end{equation}
This yields the following easy-to-implement threshold rule:
\begin{equation}
y_{nt}^{j}=
\begin{cases} 
1 & \text{if } \lambda_{nt}^j>0\,\,\text{ and } \mu_{nt}\bar o_{nt}+\xi_t \bar h_{nt}< w_{n}^j \\
0       & \text{otherwise.}
\end{cases} \label{eq:sub1}
\end{equation}
which is a deterministic decision that offloads (or not) all requests of each device (at each $t$). Then we improve the current value of $V_t(\bm{\mu})$ by updating the dual variables:
\begin{equation}
\mu_{n,t+1}=\Big[\mu_{nt}\!+\!\alpha\big(\sum_{j=1}^M\bar o_{nt}\bar \lambda_{nt}^jy_{nt}^{j}\!-\!\bar B_{nt}\big)\Big]^+,\,\, \forall\,n, \label{eq:sub2}
\end{equation}
\begin{equation}
\xi_{t+1}=\Big[\xi_t\!+\!\alpha\big(\sum_{n=1}^N\sum_{j=1}^M\bar h_{nt}\bar \lambda_{nt}^jy_{nt}^{j}-\bar H_t\big)\Big]^+, \label{eq:sub3}
\end{equation}
where $\alpha\!>\!0$ is the update step size, and return to \eqref{eq:sub0}. 

\begin{algorithm}
	\caption{OnAlgo}
	\label{OnAlgo}
	\begin{algorithmic}[1]
		\STATE \textbf{Initialization:} $t=0, \xi_0\!=\!0, \bm \mu_0\!=\!\bm 0, \bm y\!=\!\bm 0$ 
		\WHILE {True}
		\FOR {\textbf{each} device $n\in\mathcal{N}$}
		\STATE Receive objects $\mathcal{S}_{nt}=\{s_{nt}\}$;
		\STATE $\hat{\phi}_{nt},\, \sigma_{nt}  \leftarrow Q_n\big(J_n(s_{nt})\big),\,\,\forall s_{nt}\in\mathcal{S}_{nt}$
		\STATE Calculate $w_{nt}$ through \eqref{eq:accuracy};
		\STATE Observe $o_{nt}$, $h_{nt}$, $B_{nt}$ and calculate $\bar o_{nt}, \bar h_{nt}$, $\bar B_{nt}$;
		\FOR{$j=1,\ldots, M$}
		\STATE Observe $\lambda_{nt}^j$ and calculate average $\bar \lambda_{nt}^j$ and $w_n^j$;
		\STATE Decide ${y}_n^{j}$ by using \eqref{eq:sub1};
		\ENDFOR
		\STATE Update $\mu_{n,t+1}$ using \eqref{eq:sub2};
		\STATE Send averages $\bar \lambda_{nt}^j, \forall j$, to cloudlet;
		\ENDFOR	
		\STATE \underline{Cloudlet:}
		\STATE Compute tasks and receive $\bar \lambda_{nt}^{j}, \forall n$;
		\STATE Observe $H_t$ and calculate $\bar H_t$;
		\STATE Update $\xi_{t+1}$ using \eqref{eq:sub3}, and send it to devices;
		\STATE $t \leftarrow t+1$;		
		\ENDWHILE
	\end{algorithmic}
\end{algorithm}


The detailed steps that implement our online policy are as follows (with reference to \emph{OnAlgo}, Algorithm 1). Each device $n$ receives a group of objects $\mathcal{S}_{nt}$ in slot $t$ and uses its classifier to predict their classes, and the predictor to estimate the expected offloading gains (Steps 4-6). They update their statistics (step 7) and compare the expected benefits with the outsourcing costs (Step 10). Finally, they update their local dual variable for the power constraint violation (Step 12). The cloudlet classifies the received objects (Step 16) and updates its parameter estimates (Step 17) and its congestion (Step 18), which is sent to the devices. 

%
%
%
%
%
%
%
%
%
%
%
%
%
%
%
%
\vspace{-1mm}
\subsection{Performance Analysis}


The gist of our approach is that, as time evolves, the sequence of problems $\{\mathbb P_2(t)\}_t$ approaches our initial problem $\mathbb P_1$. This is true under the following mild assumption.
\begin{assumption} 
	\vspace{-1mm}	
	The perturbations of the system parameters are independent to each other, uniformly bounded, and their averages converge, e.g., $\lim_{t\rightarrow \infty}\bar B_{nt}=B_n$.
	\vspace{-1mm}
\end{assumption}
\noindent Under this assumption it is easy to see that it holds:
\begin{equation}
 \lim_{t\rightarrow\infty}\delta_t(\bm y)\!=\!0,\,\,\,\,\,\,\,\,\lim_{t\rightarrow \infty}\bm y^\top\epsilon_t\!=\!0, \,\,\,\,\,\,\forall\, \bm y. \notag 
\end{equation} 
Furthermore, note that due to boundedness of the parameters and $y_n^j\in[0,1], \forall n,j$ we have that:

\begin{equation}
	\norm{ g(\bm y)}_2\leq \sigma_{g},\,\,\,\,\,\,\,\,\, \norm{ \delta_t(\bm y) }_2\leq \sigma_{\delta_t}, 
	\forall t, \label{eq:var}
\end{equation}
and using Minkowski's inequality, we get the bound:
\begin{equation}
\norm{\bar g_t(\bm y)}=\norm{ g(\bm y) +\delta_t(\bm y) }_2\leq \sigma_{g} + \sigma_{\delta_t}.
\end{equation}
It is also easy to see that $\lim_{t\rightarrow\infty}\sigma_{\delta_t}=0$. The following Theorem is our main analytical result.
\begin{theorem}
Under Assumption 1, OnAlgo ensures the following optimality and feasibility gaps:
\begin{align*}
(i) \lim_{t\rightarrow \infty} f(\bm{\bar y_t})\leq f^* + \frac{a\sigma_g^2}{2},\quad (ii) \lim_{t\rightarrow \infty} g(\bm{\bar y_t})\preceq 0, \nonumber
\end{align*}
where $\bm{\bar y_t}=\frac{1}{t}\sum_{i=1}^t \bm y_i$.
\end{theorem}

\begin{proof}
	We drop bold typeface notation here, and use subscript $i=1,\ldots,t$ to denote the $i$-th slot.  We first bound the distance of $\mu_{t+1}$ from vector $\theta\in \mathbf R^{N+1}$, i.e., $\norm{ \mu_{t+1}-\theta }_{2}^2 = $
	\begin{align}
	&\norm{ [\mu_t+a\big( g(y_t)+\delta_t(y_t)\big) ]^+ -\theta }_{2}^2 \leq \nonumber\\
	&\norm{ \mu_t-\theta }_{2}^2 \!+\!a^2\norm{ g(y_t)}_{2}^2\!+\! a^2\norm{ \delta_t(y_t)}_{2}^2 \!+\! 2a^2\delta_t(y_t)^\top g(y_t)+ \nonumber \\ 
	&2a (\mu_t\!-\!\theta)^\top\big(g(y_t)\!+\!\delta_t(y_t)\big). \label{eq:15} 
	\end{align}
	
	\underline{(i) Optimality Gap}. From the dual problem we can write:
	\begin{align}
	V(\mu^*)&\geq \frac{1}{t}\sum_{i=1}^tV(\mu_i)\geq \frac{1}{t}\sum_{i=1}^t L(y_i,\mu_i) \nonumber \\
	&=\!\frac{1}{t}\sum_{i=1}^t\Big( f(y_i)\!+\!y_i^\top\epsilon_i\!+\!\mu_i^\top\big(g(y_i)\!+\!\delta_i(y_i)\big) \Big) \nonumber \\
	&\geq f(\bar y_t)+\frac{1}{t}\sum_{i=1}^ty_i^\top\epsilon_i + \frac{1}{t}\sum_{i=1}^t\Big( \mu_i^\top\big( g(y_i)+\delta_i(y_i) \big) \Big), \label{eq:dual_expansion}
	\end{align}
	where the last inequality follows from Jensen's inequality. Now, let $\theta=0$ in \eqref{eq:15}. Using \eqref{eq:var} and the Cauchy-Swartz inequality, and by summing over all $t$ we obtain :
	\begin{align}
	\norm{ \mu_{t+1} }_2^2 &\leq \norm{ \mu_1 }_2^2 + a^2t\sigma_g^2+ a^2\sum_{i=1}^t\sigma_{\delta_t}^2 + \nonumber \\
	&2a^2\sigma_g\sum_{i=1}^t\sigma_{\delta_t} + 2a\sum_{i=1}^t \mu_i^\top \big( g(y_i)+\delta_i(y_i) \big). \nonumber
	\end{align}
	Dropping the non-negative term $\norm{ \mu_{t+1} }_2^2$, dividing by $2at$, setting $\mu_1=0$, and rearranging terms, yields:
	\begin{align*}
	-\frac{1}{t}\sum_{i=1}^t\mu_i^\top\big(g(y_i)\!+\!\delta_i(y_i) \big) &\leq \frac{a\sigma_g^2}{2}\!+\!\frac{a}{2t}\sum_{i=1}^t\sigma_{\delta_i}^2\!+\!\frac{a\sigma_g}{t}\sum_{i=1}^t\sigma_{\delta_i}.
	\end{align*}
	Using the fact that $V(\mu^*)=f^*$, and combining the above with \eqref{eq:dual_expansion}, we obtain:
	\begin{align*}
	f(\bar y_t)\!-\!f^* &\leq -\frac{1}{t}\sum_{i=1}^ty_i^\top\epsilon_i\!+\! \frac{a\sigma_g^2}{2}\!+\! \frac{a}{2t}\sum_{i=1}^t\sigma_{\delta_i}^2\!+\!\frac{a\sigma_g}{t}\sum_{i=1}^t\sigma_{\delta_i}.
	\end{align*}
	All sums have diminishing terms and divided by $t$, hence converge to $0$. Thus, we obtained the first part of the theorem.
	
	\underline{(ii) Constraint Violation}. If we apply recursively the dual variable update rule, we obtain:
	\begin{align*}
	\mu_{t+1}\!=\!\Big[ \mu_t\!+\!a\big( g(y_t)\!+\!\delta_t(y_t) \big) \Big]^+\!\succeq\!\mu_1\!+\!a\sum_{i=1}^t\big( g(y_i)\!+\!\delta_i(y_i) \big).
	\end{align*}
	Setting $\mu_1=0$, dividing by $at$, and using Jensen's inequality for $g(\cdot)$, we get:
	\begin{align}
	g(\bar y_t) + \frac{1}{t}\sum_{i=1}^t\delta_i(y_i) \preceq \frac{\mu_{t+1}}{at}. \label{eq:constraint-bound}
	\end{align}
	The second term of the LHS converges to zero as $t\rightarrow \infty$. Our claim holds if the same is true for the RHS. Indeed, this is the case assuming the existence of a Slater vector, and the boundedness of the set of dual variables (see~\cite{nedic-subgrad-siam,appendix}).
\end{proof}



\noindent The theorem shows that OnAlgo asymptotically achieves zero feasibility gap (no constraint violation), and a fixed optimality gap that can be made arbitrarily small by tuning the step size. 


\section{Implementation and Evaluation} \label{sec:evaluation}


\subsection{Experimentation Setup and Initial Measurements}
 
\subsubsection{Testbed and Measurements} We used 4 Raspberry Pis (RPs) as end-nodes, placed in different distances from a laptop (cloudlet). We used a Monsoon monitor for the energy measurements, and Python libraries and TensorFlow for the classifiers.\footnote{We used \emph{vanilla} versions of the classifiers to facilitate observation of the results. The memory footprint of NNs can be made smaller \cite{NN-compression} but this might affect their performance. Our analysis is orthogonal to such interventions.} We first measured the average power consumption when RPs transmit data to the cloudlet with different rates, and then fitted a linear regression model that estimates the consumed power as a function of $r$. This model is used by OnAlgo to estimate the energy cost for each transmitted image, given the data rate in each slot (which might differ for the RPs). Also, we measured the average computing costs ($h_n,h_0$ cycles/task) of the classification tasks, to be used in simulations. For more details on the setup, see \cite{appendix}.

\subsubsection{Data Sets and Classifiers} We use two well-known datasets: \emph{(i)} {MNIST}~\cite{lecun1998gradient} which consists of $28\!\times\!28$ pixel handwritten digits, and includes $60$K training and $10$K test examples; \emph{(ii)} {CIFAR-10}~\cite{Krizh2009Learn} with $50$K training and $10$K test examples of $32\!\times\!32$ color images of $10$ classes. We used two classifiers, the normalized-distance weighted \textit{k}-nearest neighbors (KNN)~\cite{dudani1976distance}, and the more sophisticated Convolutional Neural Network (CNN) implemented with TensorFlow \cite{tensorflow}. They output a vector with the probabilities that the object belongs to each class. These classifiers have different performance and resource needs, hence allow us to build diverse experiments. The predictors are trained with labeled images and the outputs of the local ($f_n$) and cloudlet ($f_0$) classifiers. These are the independent variables in our regression model that estimates $\phi_{nt}$ (dependent variables). Recall that the latter are calculated using \eqref{eq:predictor-gain}, where we additionally use that $w_{nt}=d_0(s_{nt})$ if device $n$ has given a wrong classification and $w_{nt}=-d_0(s_{nt})$ if the cloudlet is mistaken. 

\begin{figure*}[h]
	\centering
	\begin{subfigure}[b]{0.23\linewidth}
		\includegraphics[scale=0.26]{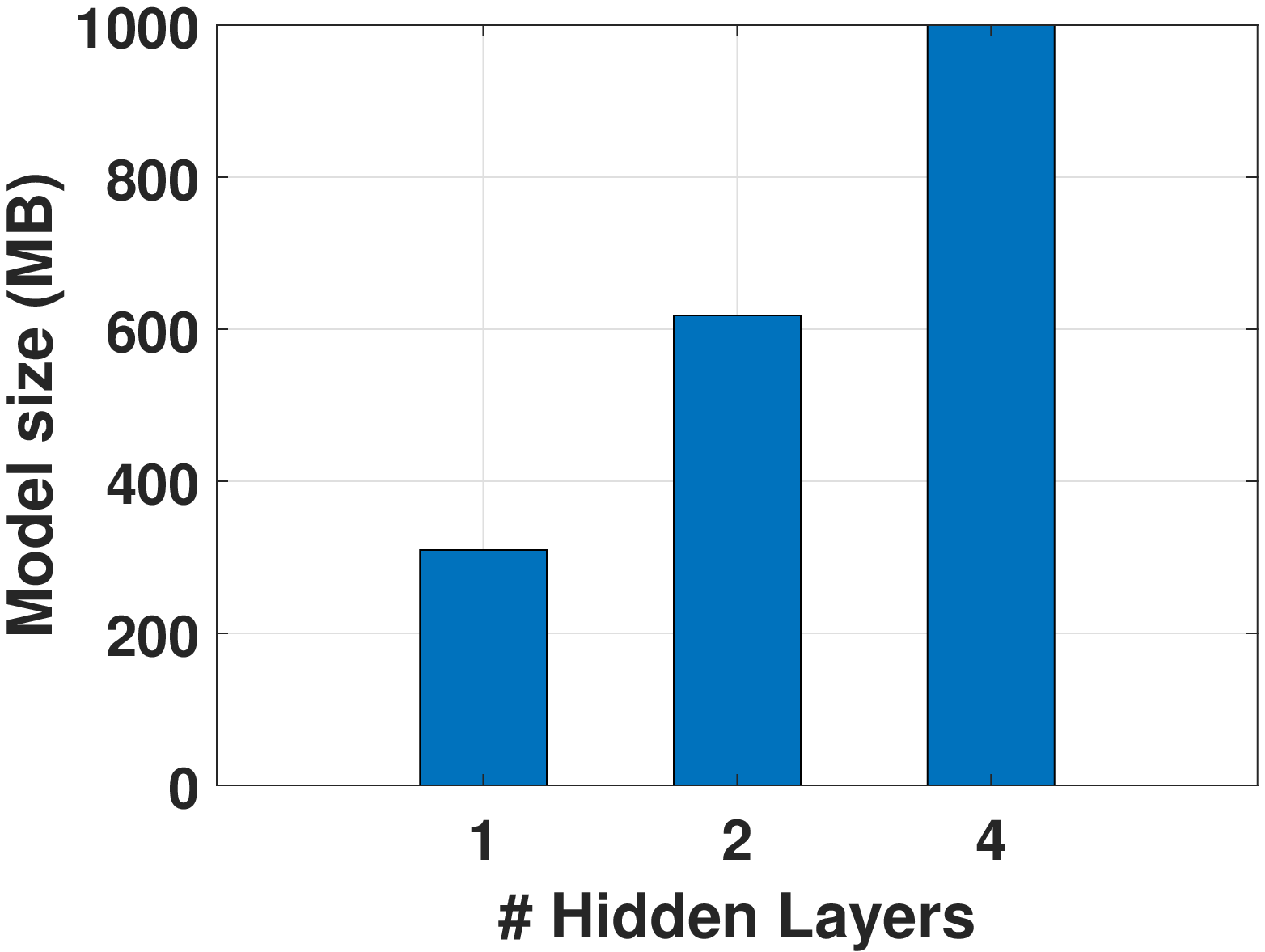}
		\caption{Memory usage of CNN}
	\end{subfigure}
	~
	\begin{subfigure}[b]{0.23\linewidth}
		\includegraphics[scale=0.26]{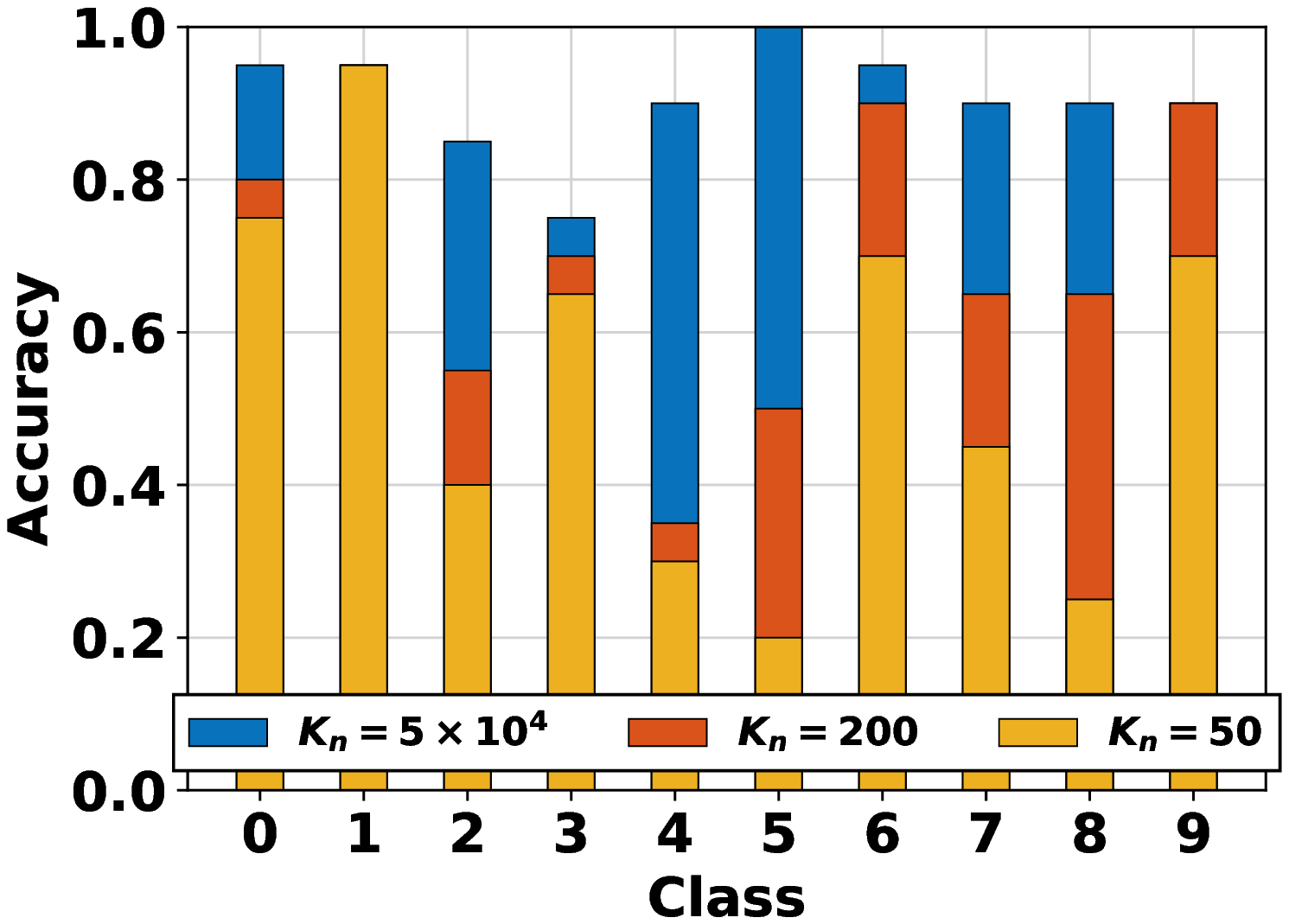}
		\caption{KNN on MNIST}
	\end{subfigure}
	~
	\begin{subfigure}[b]{0.23\linewidth}
		\includegraphics[scale=0.27]{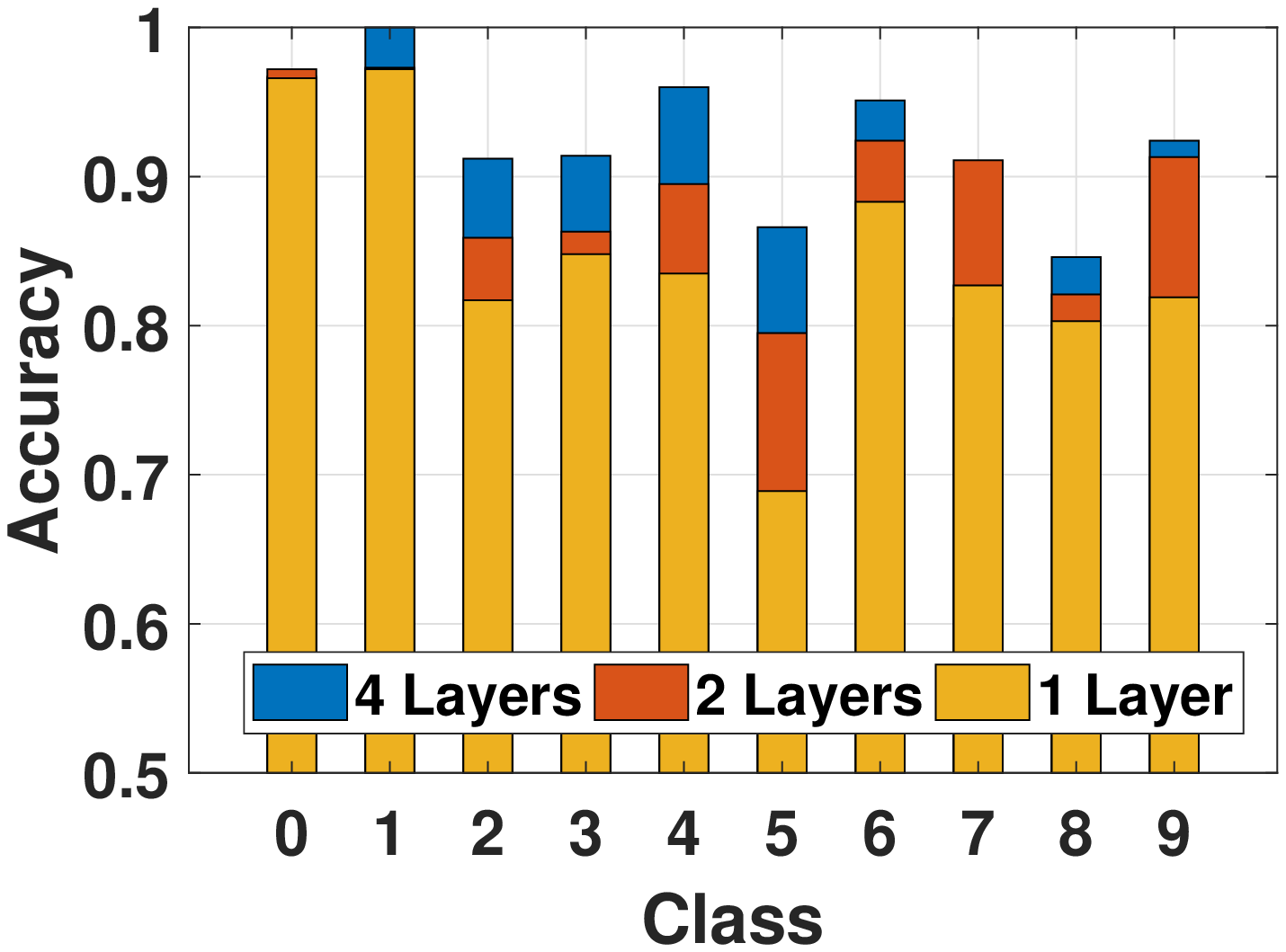}
		\caption{CNN on MNIST}
	\end{subfigure}
	~	
	\begin{subfigure}[b]{0.23\linewidth}
		\includegraphics[scale=0.275]{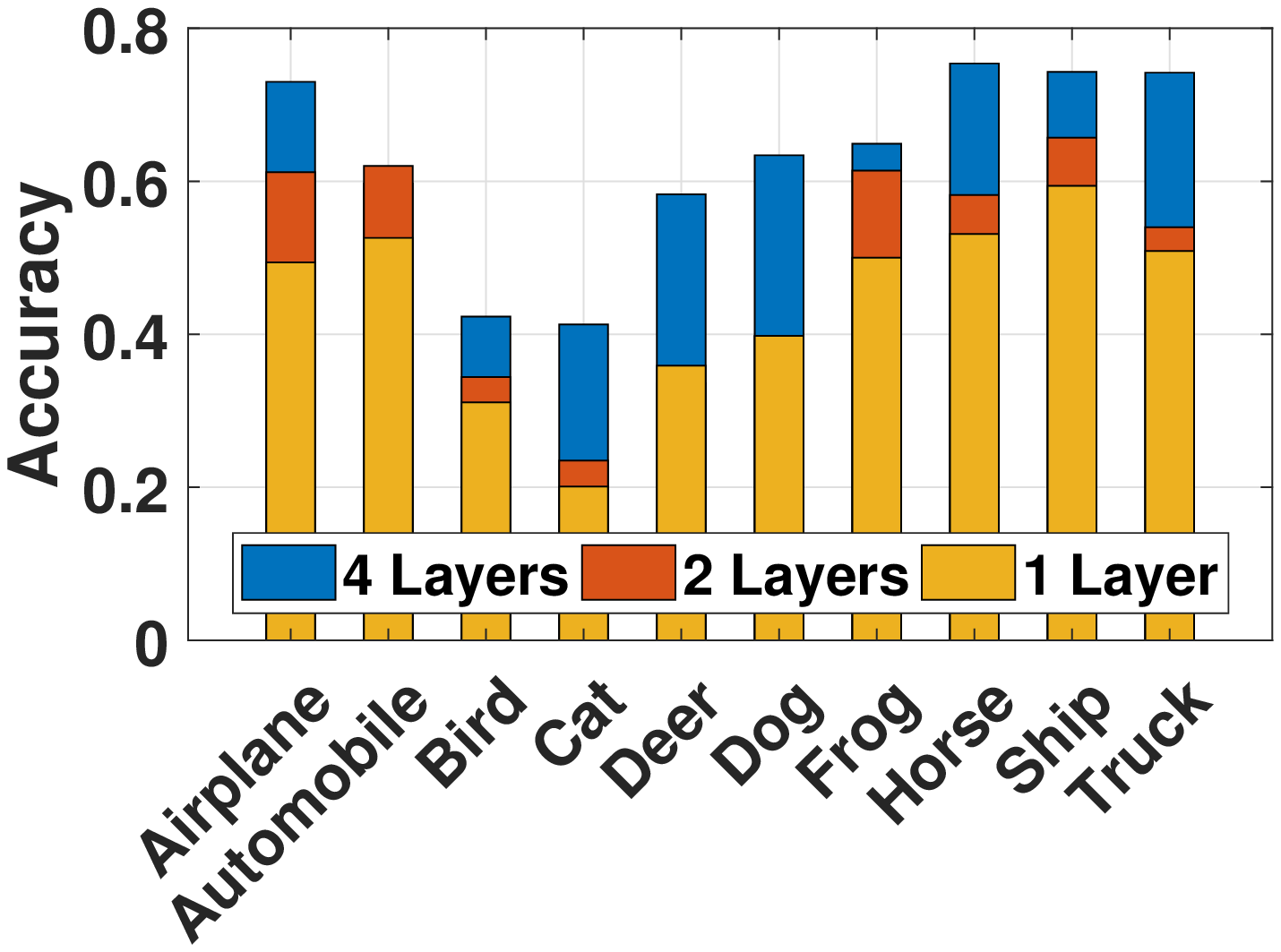}
		\caption{CNN on CIFAR}
	\end{subfigure}
	\caption{CNN memory usage vs number of layers, and accuracy of MNIST and CIFAR-10 for KNN and CNN classifiers.} 
	\label{fig:AccuracyPerClass}
	\vspace{-2mm}
\end{figure*}

\subsubsection{Benchmarks} We compare OnAlgo with two algorithms. The \emph{Accuracy-Threshold Offloading} (ATO) algorithm, where a task is offloaded when the confidence of the local classifier is below a threshold, without considering the resource consumption. And the \emph{Resource-Consumption Offloading} (RCO) algorithm, where a task is offloaded when there is enough energy, without considering the expected classification improvement.

\subsubsection{Limitations of Mobile Devices} We used our testbed to verify that these small resource-footprint devices require the assistance of a cloudlet. Our findings are in line with previous studies, e.g., \cite{DDNN_1}. The performance of a CNN model increases with the number of layers. We find that, even with $4$ layers, a CNN trained for CIFAR has $1$GB size and hence cannot be stored in the RPs (see Fig.~\ref{fig:AccuracyPerClass}a). Similar conclusions hold for the KNN classifier that needs to locally store all training samples. Clearly, despite the successful efforts to reduce the size of ML models by, e.g. using compression~\cite{NN-compression}; the increasingly complex analytics and the small form-factor of devices will continue to raise the local versus cloudlet execution trade off.

\begin{figure}[t!]
	\centering
	\begin{subfigure}[b]{0.47\linewidth}
		\centering
		\includegraphics[scale=0.22]{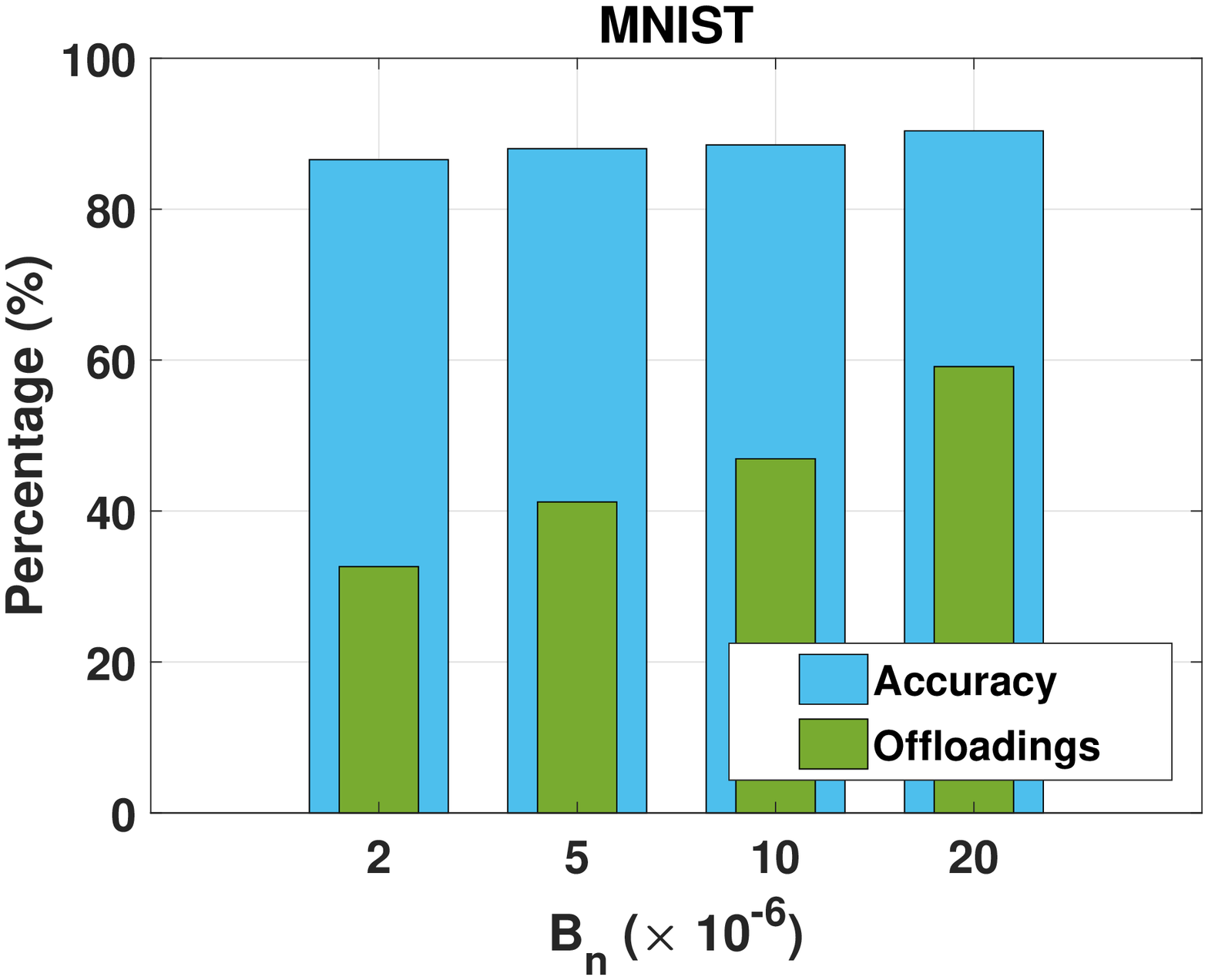}
	\end{subfigure}
	~
	\begin{subfigure}[b]{0.47\linewidth}
		\centering
		\includegraphics[scale=0.22]{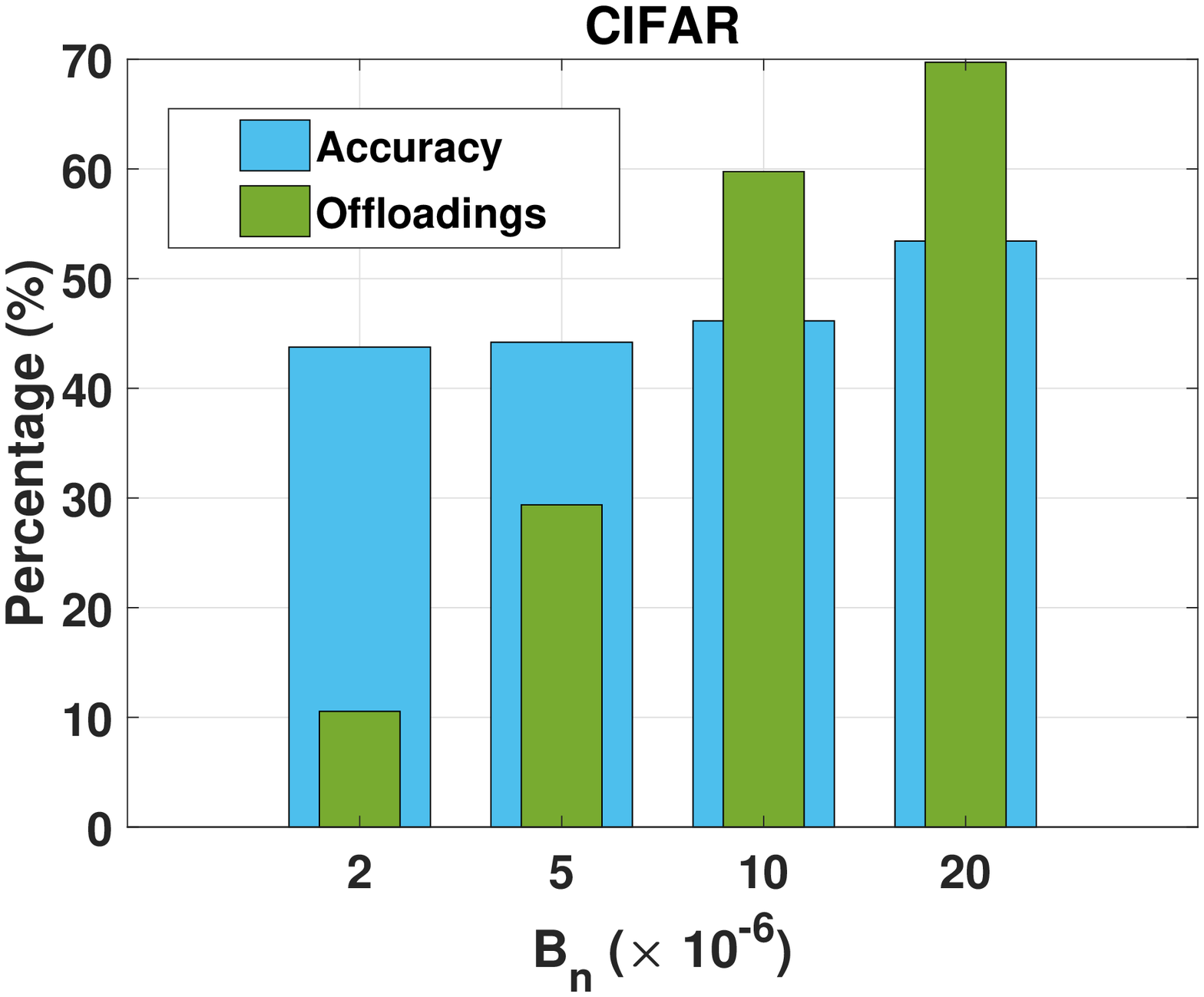}
	\end{subfigure}
	\vspace{-1mm}
	\caption{Average accuracy and outsourcing of OnAlgo for different power budgets, on MNIST (left) and CIFAR (right).}
	\label{fig:Resource_results}
	\vspace{-2mm}
\end{figure}

\subsubsection{Classifier Assessment} In Fig.~\ref{fig:AccuracyPerClass}b we see that the accuracy (ratio of successful over total predictions) of the KNN classifier improves with the size $K_n$ of labeled data. Figure~\ref{fig:AccuracyPerClass}c presents the accuracy gains for CNN as more hidden layers are added. The gains are higher (up to 20\%) for the digits that are difficult to recognize, e.g., $4$ and $5$. Fig.~\ref{fig:AccuracyPerClass}d shows the CNN performance on CIFAR, which is lower as this is a more complex dataset (colored images, etc.). Overall, we see that the classifier performance depends on the algorithm (KNN, CNN), the settings (datasets, layers), and the objects.

\subsection{Performance Evaluation}

\subsubsection{Resource Availability Impact}  Fig. \ref{fig:Resource_results} shows the average accuracy and fraction of requests offloaded to the cloudlet with OnAlgo when we vary their power budget. As $B_n$ increases there are more opportunities to use the cloudlet (4-layer CNN) and obtain more accurate classifications than the local classifier (1-layer CNN). Furthermore, Fig. \ref{fig:AccuracyPerClass}(c-d) show that MNIST is easier to classify and the gains of using a better classifier are smaller than with CIFAR. Hence, as $B_n$ increases in Fig. \ref{fig:Resource_results} the ratio of offloaded tasks increases at a faster pace with CIFAR than with MNIST.

\subsubsection{Comparison with Benchmarks} We compare OnAlgo to ATO and RCO. No-offloading (NO) serves as a baseline for these algorithms in Fig. \ref{fig:comparison_1}. To ensure a realistic comparison, we set the rule for all algorithms that the cloudlet will not serve any task if the computing capacity constraint is violated. For RCO, the availability of energy is determined by computing the running average consumption at each device during the experiment. We employ two testbed scenarios, and a simulation with larger number of devices.

\textbf{Scenario 1:} \textit{Low accuracy improvement; high resources}. We set\footnote{We have explicitly set a small power budget so as to highlight the impact of power constraints on the system performance; higher power budgets will still be a bottleneck for higher task request rates or images of larger size.} $B_n\!=\!0.01mW$ and $H\!=\!2GHz$ allowing the devices to offload many tasks, and the cloudlet to serve most of them; and used MNIST (has small improvement). We demonstrate the average accuracy and power consumption in Fig. \ref{fig:comparison_1}a, where we see that OnAlgo outperforms both ATO and RCO by $5\%$. Regarding power consumption, ATO achieves the best result since it gets high enough confidence on its local classifier (rarely offloads). RCO however, offloads almost every task as it has enough resources and does not refrain even when improvement is low. The reason it achieves lower accuracy than onAlgo is that it does not offload intelligently, and gets denied when the computing constraint is violated.

\begin{figure}[t!]
	\centering
	\begin{subfigure}[b]{0.47\linewidth}
		\centering
		\includegraphics[scale=0.2]{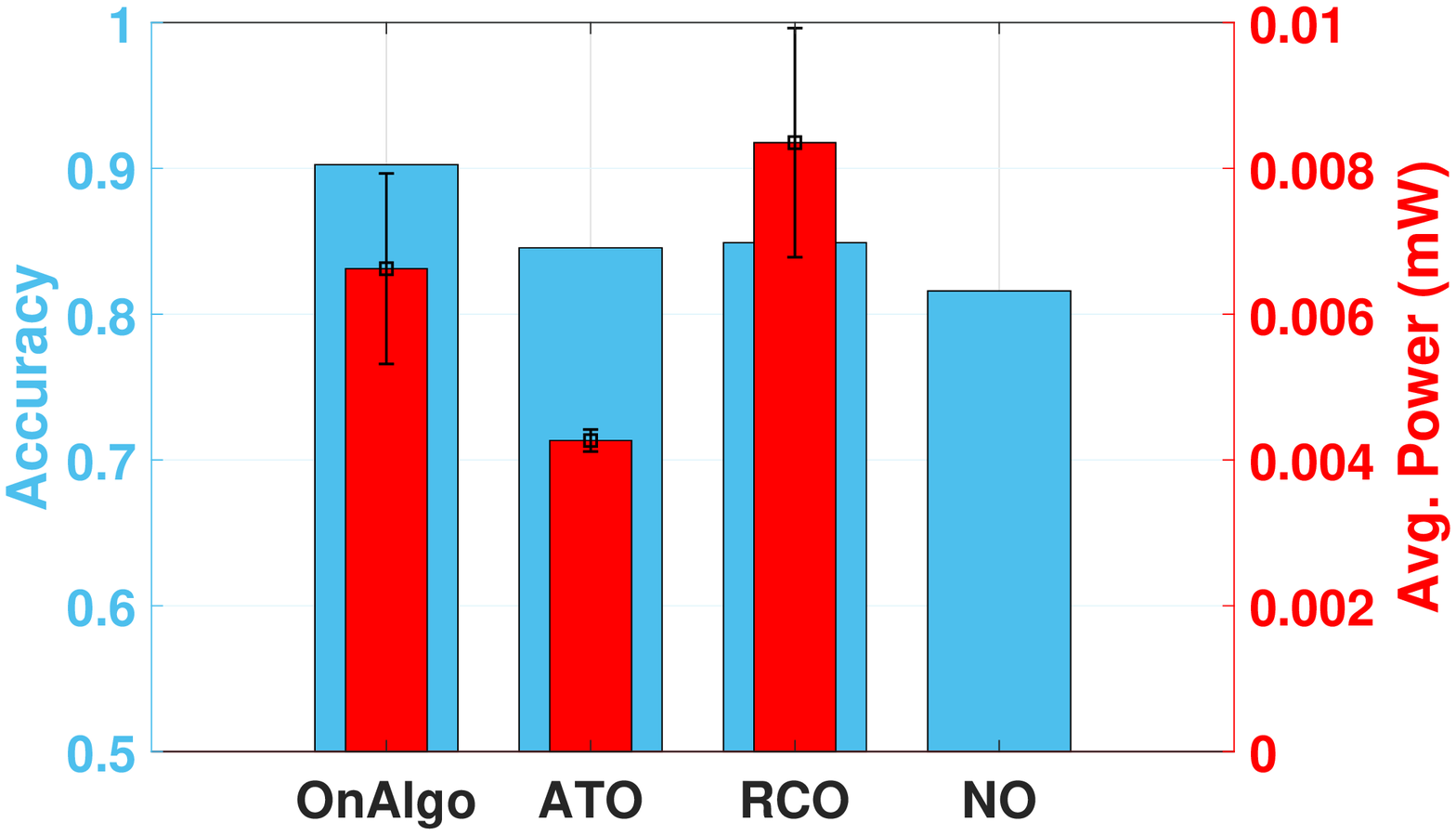}
		\caption{Scenario 1}
	\end{subfigure}
	~
	\begin{subfigure}[b]{0.47\linewidth}
		\centering
		\includegraphics[scale=0.2]{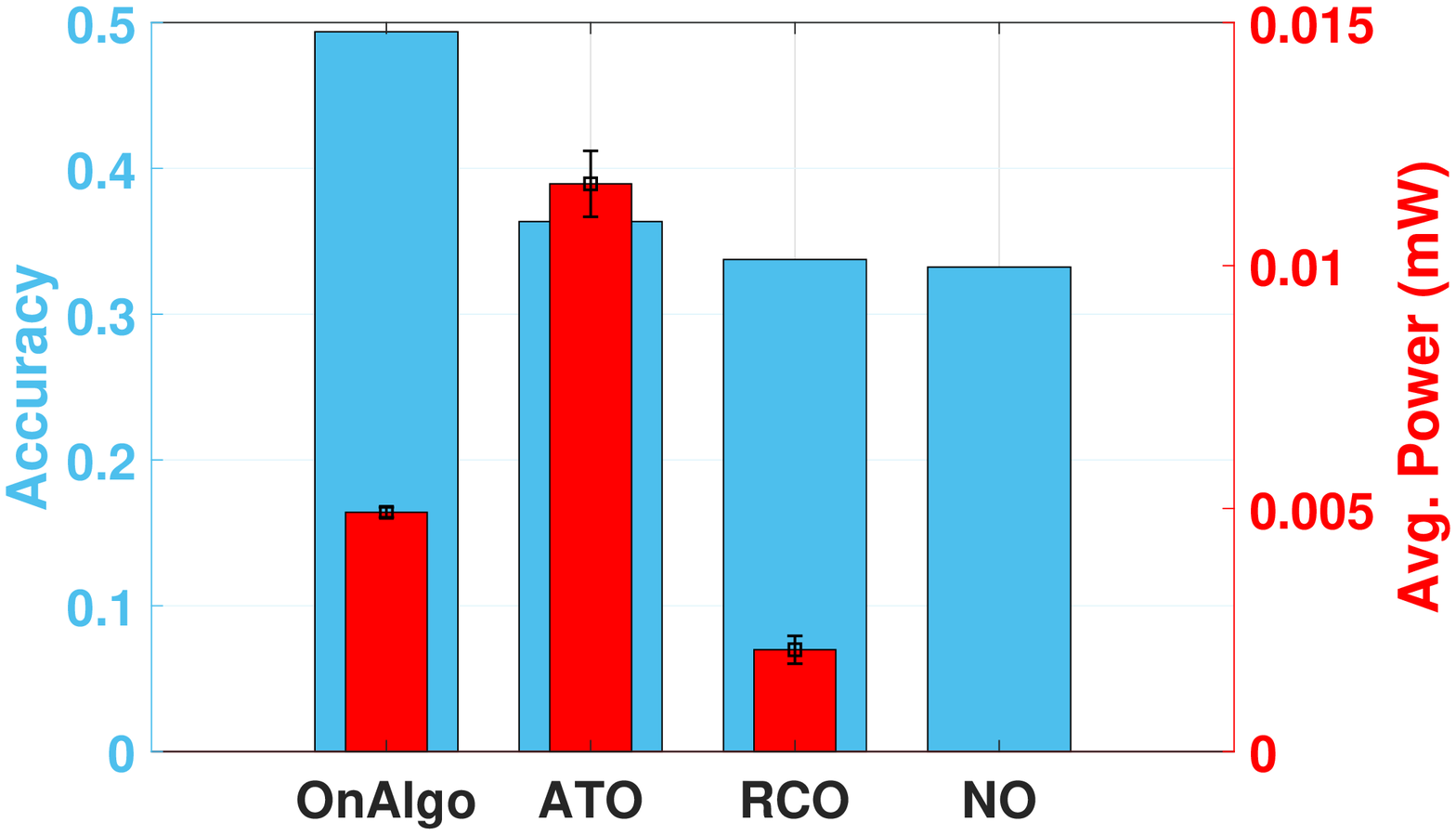}
		\caption{Scenario 2}
	\end{subfigure}
	\vspace{-1mm}
	\caption{Performance comparison of the offloading algorithms.}
	\label{fig:comparison_1}
	\vspace{-2mm}
\end{figure}

\textbf{Scenario 2:} \textit{High accuracy improvement; low resources}. We set $B_n\!=\!0.005mW$ and $H\!=\!200MHz$ not allowing many offloadings and cloudlet classifications. We used the CIFAR dataset which has a large performance difference between local and cloudlet classifiers. We see from Fig.~\ref{fig:comparison_1}b that OnAlgo achieves 28\%-32\% higher accuracy than both competing algorithms. RCO is constrained to very few offloadings due to the limited power budget, while ATO is resource-oblivious and offloads tasks regardless of the cloudlet's capacity. This results in many denied offloadings that reduce ATO's accuracy and unnecessarily increase the power consumption. OnAlgo consumes 60\% less power than ATO as it frequently offloads its low-confidence tasks. 

\textbf{Scenario 3:} \textit{Large number of users}. Finally, we simulated the algorithms for a large number of users while using the experimentally measured parameters. We observe in Fig. \ref{fig:Simulation}a that the accuracy gradually drops (for all algorithms) since now a smaller percentage of the tasks can be served by the cloudlet. OnAlgo constantly outperforms both ATO and RCO by about $10\%-25\%$ since it adapts to the available resources. This is more evident in Fig.~\ref{fig:Simulation}b that shows the fast-increasing energy cost of the two benchmark algorithms, as they either offload tasks that do not improve the performance, or offload tasks while the cloudlet is already congested (these tasks are dropped and energy is wasted). Power consumption of OnAlgo is up to $50\%$ less than that of RCO.

\begin{figure}[t]
	\centering
	\begin{subfigure}[b]{0.47\linewidth}
		\centering
		\includegraphics[scale=0.24]{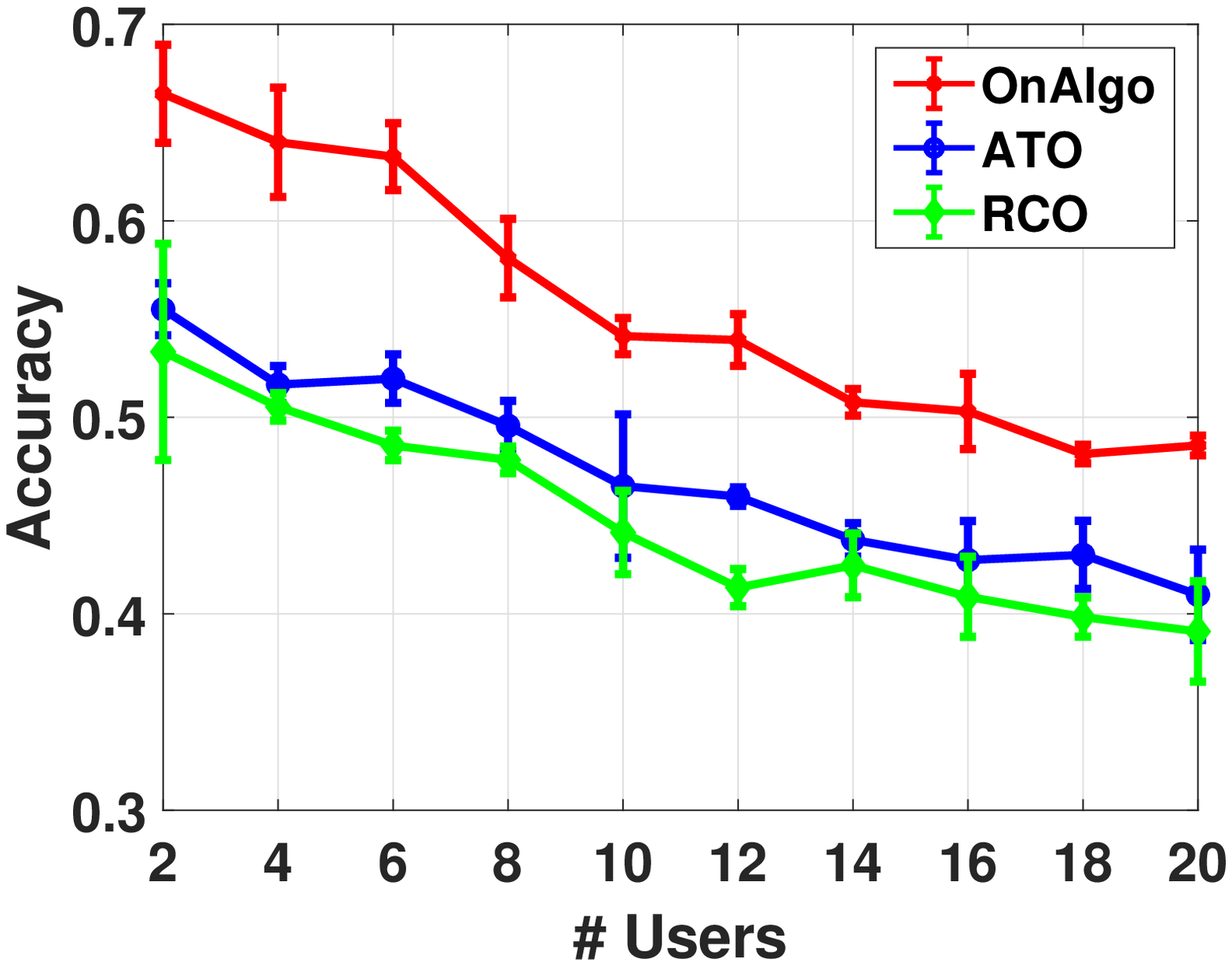}
		\caption{Accuracy Comparison}
	\end{subfigure}
	\begin{subfigure}[b]{0.47\linewidth}
		\centering
		\includegraphics[scale=0.24]{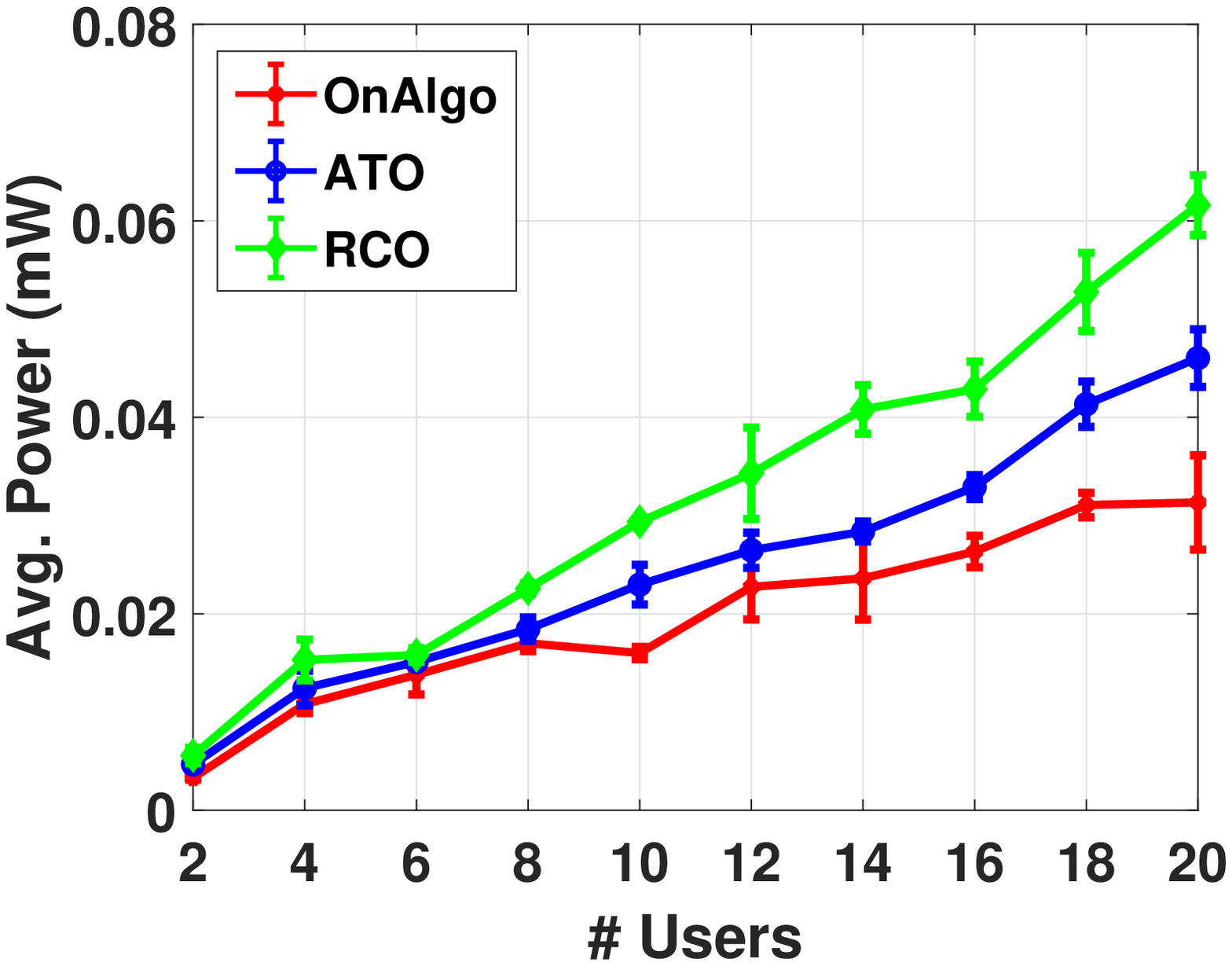}
		\caption{Power Cost Comparison}
	\end{subfigure}
	\vspace{-1.5mm}
	\caption{Simulation results for increasing number of users on the CIFAR dataset. $B_n = 0.02\ mW, H = 2\ GHz$.}
	\label{fig:Simulation}
	\vspace{-2mm}
\end{figure}

\subsubsection{Convergence of OnAlgo} Fig.~\ref{fig:convergence} presents the convergence of OnAlgo for different step sizes $\alpha$. Based on the system parameters the bound given by Theorem 1 is approximately 0.01, 0.2 and 1 for the three $\alpha$ values of Fig.~\ref{fig:convergence}. These are satisfied by the solution of OnAlgo in less than 300 iterations as observed in Fig.~\ref{fig:convergence}a. The convergence is faster for larger $\alpha$, which however is achieved at the cost of smaller convergence accuracy. The constraint violation bound is also respected as shown in Fig.~\ref{fig:convergence}b with the constraints being violated more often for small $\alpha$ in the beginning, but improving as $T$ increases.

\begin{figure}[t!]
	\centering
	\begin{subfigure}[b]{0.47\linewidth}
		\centering
		\includegraphics[scale=0.3]{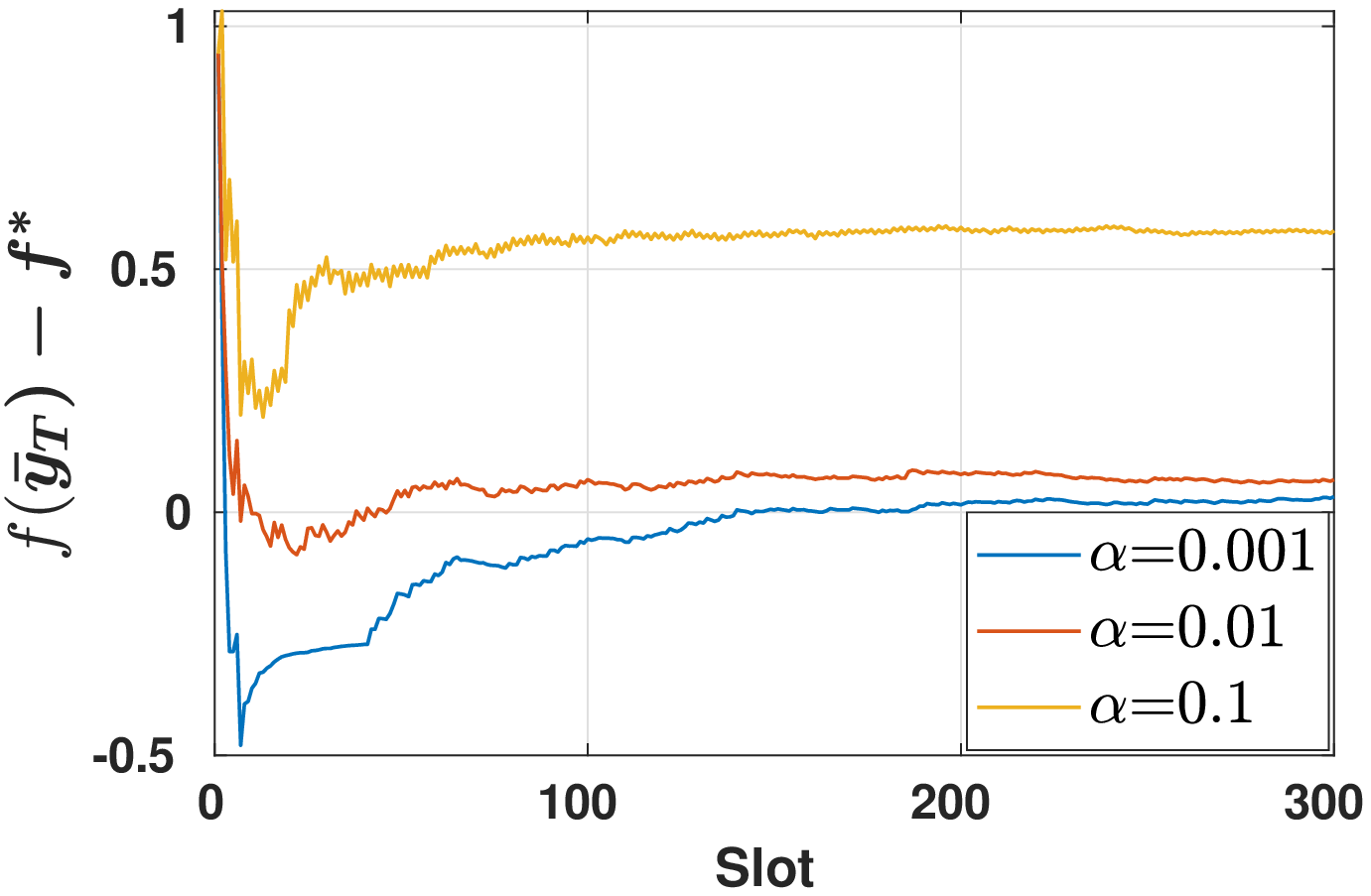}
		\caption{Optimality gap}
	\end{subfigure}
	~
	\begin{subfigure}[b]{0.47\linewidth}
		\centering
		\includegraphics[scale=0.3]{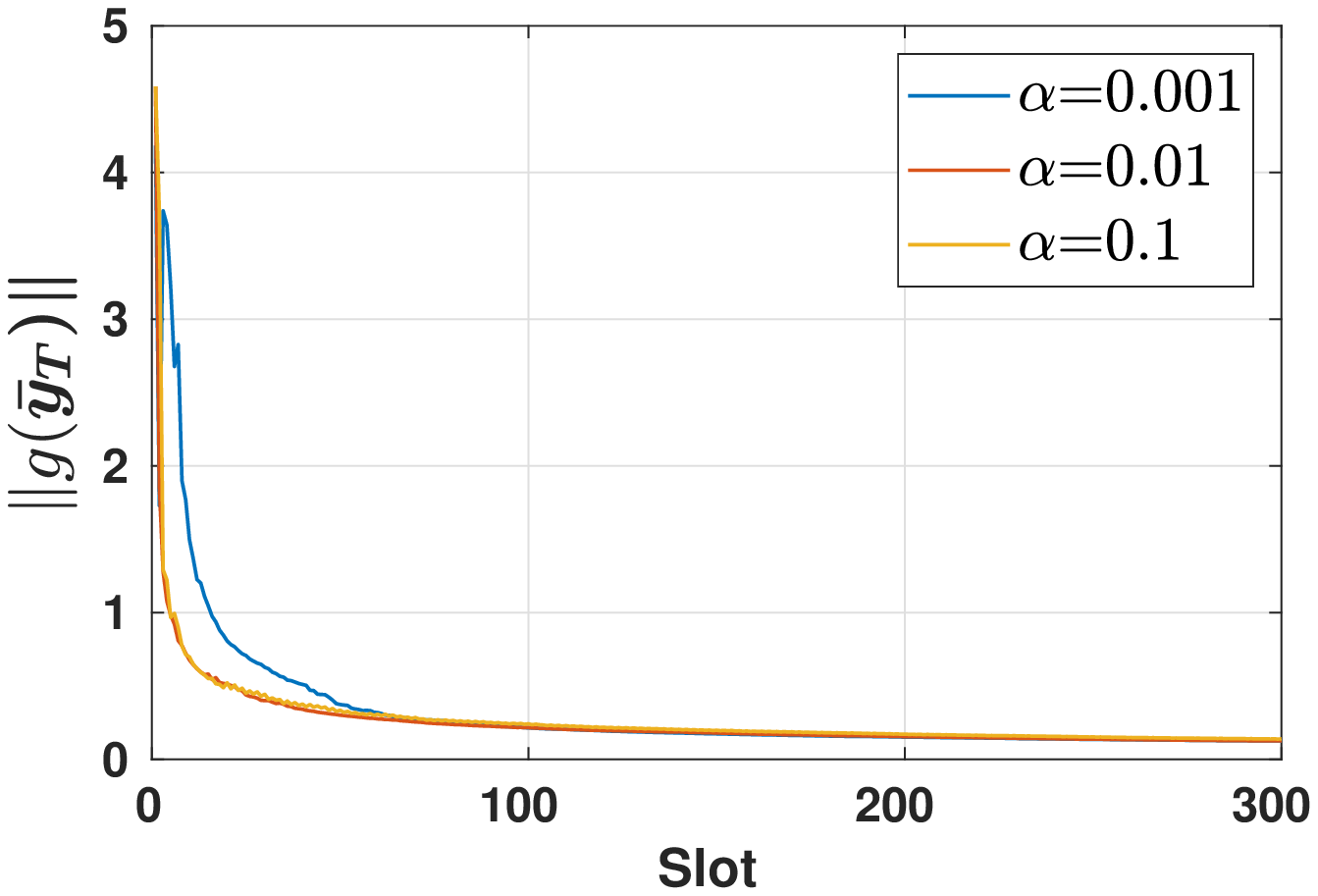}
		\caption{Constraint violation}
	\end{subfigure}
	\caption{The convergence properties of OnAlgo when $M=6, N=5$. Note that constraints are eventually satisfied, with some of them in a strict fashion (hence the norm is not zero).}
	\label{fig:convergence}
	\vspace{-2mm}
\end{figure}

\section{Related Work} \label{sec:related}
\textbf{Edge \& Distributed Computing}. Most solutions partition compute-intense mobile applications and offload them to the cloud~\cite{CloneCloud}; a solution that is unfit to enable low-latency applications. Cloudlets on the other hand, achieve lower delay~\cite{Cloudlets} but have limited serving capacity, hence there is a need for an intelligent offloading strategy that we propose here.  
%
Previous works consider simple performance criteria, such as reducing computation loads~\cite{Misco}, or power consumption~\cite{Selective_offloading} and focus on the architecture design. Also, Mobistreams~\cite{MobiStreams} and Swing~\cite{Swing} focus on collaborative data stream computations. The above systems either do not optimize the offloading policy, or use heuristics that do not cater for task accuracy.


\textbf{Mobile and IoT Analytics}. The importance of analytics has motivated the design of wireless systems that can execute such tasks. For instance, \cite{vrarDeep,deepdecision} tailor deep neural networks for execution in mobile devices,
while \cite{mobiqor} and \cite{hetero-edge} minimize the execution time for known system parameters and task loads. 
Finally, \cite{Li_19,Live_Analytics,Smart_cities} leverage the edge architecture to effectively execute analytics for IoT devices.
The plethora of such system proposals, underlines the necessity for our \emph{online decision framework} that provides optimal execution of analytics.

\textbf{Optimization of Analytics}. Prior works in computation offloading focus on different metrics such as number of served requests,  \cite{xu-ToN, letaief-edge-tutorial}, and hence are not applicable here. In our previous work~\cite{wiopt}, we proposed a \emph{static} collaborative optimization framework, which does not employ predictions nor accounts for computation constraints. Other works, e.g.~\cite{deepdecision} either rely on heuristics or assume static systems and known requests. Clearly, these assumptions are invalid for many practical cases where system parameters not only vary with time, but often do not follow i.i.d. processes. This renders the application of max-weight type of policies \cite{tassiulas-book} inefficient. Our approach is fundamentally different and leads to an online robust algorithm and
is inspired by dual averaging and primal recovery algorithms for static problems, see \cite{nedic-subgrad-siam}. 

\textbf{Improvement of ML Models}. Clearly, despite the efforts to improve the execution of analytics at small devices, e.g., by residual learning or compression \cite{NN-compression}, the trade off between local low-accuracy and cloudlet high-accuracy execution is still important due to the increasing number and complexity of these tasks. This observation has spurred efforts for designing fast multi-tier (cloud to edge) deep neural networks \cite{DDNN_1} and for dynamic model selection \cite{dnn-dynamic}, among others. These works are orthogonal to our approach and can be directly incorporated in our framework. 


\section{Conclusions} \label{sec:conclusions}

We propose the idea of improving the execution of data analytics at IoT devices with more robust instances running at cloudlets. The key feature of our proposal is a dynamic and distributed algorithm that makes the outsourcing decisions based on the expected performance improvement, and the available resources at the devices and cloudlet. 
The proposed algorithm achieves near-optimal performance in a deterministic fashion, and under minimal assumptions about the system behavior. 
This makes it ideal for the problem at hand where, the stochastic effects (e.g., expected accuracy gains) have unknown mean values and possibly non-i.i.d. behavior. 


\section*{Acknowledgments}
This publication has emanated from research supported in part by SFI research grants 17/CDA/4760, 16/IA/4610 and is co-funded under the European Regional Development Fund under Grant Number 13/RC/2077.

\bibliographystyle{IEEEtran}
\bibliography{citations}

\end{document}